\documentclass[conference,onecolumn]{IEEEtran}
\IEEEoverridecommandlockouts

\usepackage{cite}
\usepackage{algorithmic}
\usepackage{graphicx}
\usepackage{textcomp}
\usepackage{xcolor}
\usepackage{balance}
\usepackage[cmex10]{amsmath}
\usepackage{amsmath}
\usepackage{amssymb}
\usepackage{amsthm}
\usepackage{amsfonts}
\usepackage{bm}
\usepackage{xfrac}
\usepackage{empheq}
\usepackage[normalem]{ulem} 
\usepackage{soul} 
\usepackage{mathtools}
\newtheorem{theorem}{Theorem}

\newtheorem{example}{Example}
\newtheorem{proposition}{Proposition}
\newtheorem{lemma}{Lemma}

\newtheorem{corollary}{Corollary}
\newtheorem{remark}{Remark}
\theoremstyle{definition}

\def\BibTeX{{\rm B\kern-.05em{\sc i\kern-.025em b}\kern-.08em
    T\kern-.1667em\lower.7ex\hbox{E}\kern-.125emX}}
\begin{document}

\title{Privacy-Utility Trade-offs Under Multi-Level Point-Wise Leakage Constraints
}

\author{\IEEEauthorblockN{1\textsuperscript{st} Amirreza Zamani}
\IEEEauthorblockA{\textit{Department of Information Science} \\
\textit{and Engineering, KTH}\\
Stockholm, Sweden \\
amizam@kth.se}
\and
\IEEEauthorblockN{2\textsuperscript{nd} Parastoo Sadeghi}
\IEEEauthorblockA{\textit{School of Engineering and Technology,} \\
\textit{ UNSW}\\
Canberra, Australia \\
p.sadeghi@unsw.edu.au}
\and
\IEEEauthorblockN{3\textsuperscript{rd} Mikael Skoglund}
\IEEEauthorblockA{\textit{Department of Information Science} \\
\textit{and Engineering, KTH}\\
Stockholm, Sweden \\
skoglund@kth.se}}
\maketitle
\begin{abstract}
An information-theoretic privacy mechanism design is studied, where an agent observes useful data $Y$  which is correlated with the private data $X$. The agent wants to reveal the information to a user, hence, the agent utilizes a privacy mechanism to produce disclosed data $U$ that can be revealed. We assume that the agent has no direct access to $X$, i.e., the private data is hidden. We study privacy mechanism design that maximizes the disclosed information about $Y$, measured by the mutual information between $Y$ and $U$, while satisfying a point-wise constraint with different privacy leakage budgets. We introduce a new measure, called the \emph{multi-level point-wise leakage}, which allows us to impose different leakage levels for different realizations of $U$. In contrast to previous studies on point-wise measures, which use the same leakage level for each realization, we consider a more general scenario in which each data point can leak information up to a different threshold. 
As a result, this concept also covers cases in which some data points should not leak any information about the private data, i.e., they must satisfy perfect privacy. In other words, a combination of perfect privacy and non-zero leakage can be considered.

	When the leakage is sufficiently small, concepts from information geometry allow us to locally approximate the mutual information. We show that when the leakage matrix $P_{X|Y}$ is invertible, utilizing this approximation leads to a quadratic optimization problem that has closed-form solution under some constraints. In particular, we show that it is sufficient to consider only binary $U$ to attain the optimal utility. 
    This leads to simple privacy designs with low complexity which are based on finding the maximum singular value and singular vector of a matrix. 
    We then extend our approach to a general leakage matrix and show that the main complex privacy-utility trade-off problem can be approximated by a linear program. Finally, we discuss how to extend the permissible range of leakage.  
\end{abstract}

\begin{IEEEkeywords}
point-wise measures, multi-level leakage, information geometry, local approximation.
\end{IEEEkeywords}

\section{Introduction}
In this paper, as shown in Fig.~\ref{sys1}, an agent wants to disclose some useful information to a user. We show the useful data by a random variable (RV) $Y$, which is arbitrarily correlated with the private data denoted by RV $X$. Furthermore, RV $U$ denotes the disclosed data revealed by the agent. The agent’s goal is to design $U$ based on $Y$ so as to disclose as much information as possible about $Y$ while satisfying a point-wise privacy criterion.
 We use mutual information to measure utility and a multi-level point-wise constraint to measure privacy leakage. The multi-level point-wise leakage measure generalizes previous point-wise measures in the sense that each letter (realization) of the disclosed data $U$ is subject to a different privacy threshold. This framework also includes perfect privacy, i.e., some letters must satisfy a zero-leakage constraint. To motivate our model, we consider scenarios in which the letters of $U$ have different priorities. For instance, some letters may be more sensitive than others. As an example, in clinical studies, if a patient does not have a sensitive disease, releasing this information may not be harmful; however, if a patient does have such a disease, revealing this information can be a privacy breach. Therefore, it makes sense to require a stricter privacy condition on more sensitive outcomes of $U$ than the less sensitive outcomes of $U$.  
 

\begin{figure}[]
	\centering
	\includegraphics[width = 0.5\textwidth]{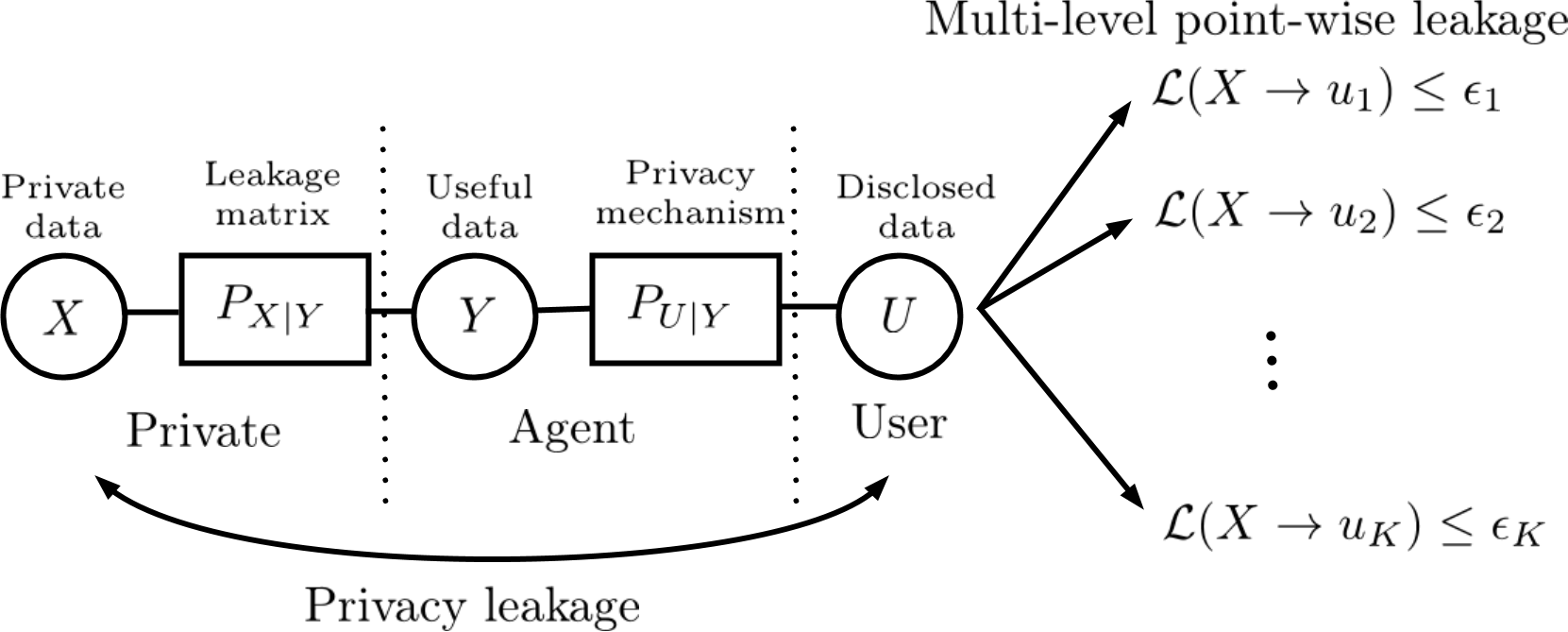}
	\caption{In this model, disclosed data $U$ is designed by a privacy mechanism that maximizes the information disclosed about $Y$ and satisfies the multi-level point-wise privacy constraint. Here, each realization of $U$ must satisfy a privacy constraint with different leakage budget. Furthermore, $\mathcal{L}(X \rightarrow u_i)$ denotes the leakage from the private data $X$ to the letter (realization) $u_i \in \mathcal{U}$.}
	\label{sys1}
\end{figure}
\subsection{Related works}
Related works on the information-theoretic privacy mechanism design can be found in \cite{borz,koala,shah,khodam,Khodam22,zarab1,zarab2,seif,duchi,kairouz,evfimievski,kairouz2015composition,barthe2013beyond,feldman2018privacy,makhdoumi, Total, Calmon1,oof,razegh,emma,ang,sep,lopuha,long,11195280}. 
In \cite{borz}, the problem of privacy-utility trade-off considering mutual information both as measures of privacy and utility is studied. Under perfect privacy assumption, it has been shown that the privacy mechanism design problem can be reduced to linear programming. Here, perfect privacy corresponds to the complete independence between $U$ and $X$, i.e., leakage from $X$ to each letter of $U$ must be zero. In \cite{koala}, \emph{secrecy by design} problem is studied under the perfect secrecy assumption. Bounds on secure decomposition have been derived using the Functional Representation Lemma. 
In \cite{shah}, the privacy problems considered in \cite{koala} are generalized by relaxing the perfect secrecy constraint and allowing some leakage. Furthermore, the bounds in \cite{shah} have been tightened in \cite{sep} by using a \emph{separation technique}.

As discussed in \cite{khodam} and \cite{Khodam22}, it may not be desirable to use average measures to quantify privacy leakage, since some data points (realizations) may leak more information than a prescribed threshold. In other words, if an adversary has access to such letters, it can infer a considerable amount of information about the sensitive data $X$. For example, average $\chi^2$ and $\ell_1$ measures were used in \cite{Calmon2} and \cite{Total}, respectively, where ``average'' corresponds to taking an average over the letters of $U$.
 On the other hand, these notions were strengthened by introducing point-wise (strong) $\chi^2$ and $\ell_1$ privacy measures in \cite{khodam} and \cite{Khodam22}, respectively. It is shown that by using concepts from information geometry, the main complex design problem can be approximated by linear algebraic techniques.
The concept of \emph{lift}, which corresponds to another point-wise measure, is studied in \cite{zarab2} and represents the likelihood ratio between the posterior and prior beliefs concerning sensitive features within a dataset.

Concepts from information geometry have been used in various network information theory problems, as well as in privacy, secrecy, and fairness design problems, to approximate complex optimization problems and derive simple designs \cite{Shashi,huang,khodam,Khodam22,shah,razegh,emma,ang,long,zamani2025fair,zamani2025fair2}. The main challenge in these problems arises from the lack of a geometric structure in the space of probability distributions. If we assume that the distributions of interest are close, KL-divergence as well as mutual information can be approximated by a weighted squared Euclidean distance. This leads to a method that allows us to approximate various design problems of interest.

Particularly, this approach has been used in \cite{Shashi, huang}, considering point-to-point channels and some specific broadcast channels. This has been used in the privacy context in \cite{khodam,Khodam22,shah,razegh,emma,long,ang}, where in \cite{ang}, mutual information is used as the privacy measure and relative entropy as the utility measure, both of which are approximated by quadratic functions. In \cite{long}, the privacy mechanisms are designed considering LIP and max-lift as the privacy leakage measures. Furthermore, it has been used in \cite{zamani2025fair2,zamani2025fair} to design fair mechanisms under bounded demographic parity and equalized odds constraints.
Furthermore, in \cite{emma}, using information geometry leads to a local approximation of the secrecy capacity over a wire-tap channel. Lastly, a similar approach has been used in \cite{ang}, where a hypothesis testing problem is considered under a privacy leakage constraint defined by bounded mutual information. In high privacy regimes, the problem is approximated.

\subsection{Contributions}
We emphasize that none of the aforementioned works study multi-level point-wise privacy constraints; that is, in all prior work, $\mathcal{L}(X \to u)$ is bounded by a single quantity $\epsilon$. As motivated above, this represents an important gap in the existing literature.
In the present work, we divide our results into two main parts as follows.
\subsubsection{Invertible leakage matrix} In the first part of the results, we assume that the leakage matrix $P_{X|Y}$ is invertible and employ a multi-level point-wise $\chi^2$-privacy constraint that extends the measure in \cite{khodam} by allowing different leakage thresholds for different realizations $u \in \mathcal{U}$. By using concepts from information geometry, we approximate the main privacy–utility trade-off problem and convert it into a simpler linear algebra problem. When $U$ is binary, the optimal solution to the approximate problem is derived in closed form.
 Furthermore, in the general case, we show that the approximate problem reduces to the case of binary $U$. In other words, for any alphabet size of $U$, we can, without loss of optimality, restrict $U$ to be binary. This implies that, in the optimal solution, only the two alphabet symbols with the highest leakage thresholds are correlated with $X$, while all other symbols do not leak any information about $X$.
\subsubsection{General leakage matrix} In the second part of the results, we first assume that the leakage matrix $P_{X|Y}$ has full row rank.
 Furthermore, we use a multi-level point-wise $\ell_1$-privacy constraint that extends \cite{Khodam22} by allowing different leakage budgets for different realizations $u \in \mathcal{U}$. 
 We show that the optimizers of the main problem lie at the extreme points of certain convex polytopes. We then use a Taylor expansion to approximate the main problem evaluated at these extreme points. Finally, we show that the resulting problem reduces to a linear program. 
 We also discuss how to generalize the framework to arbitrary leakage matrices.
\section{System model and Problem Formulation} \label{system}
Let $P_{XY}$ denote the joint distribution of discrete random variables $X$ and $Y$ defined on finite alphabets $\cal{X}$ and $\cal{Y}$. We represent $P_{XY}$ by a matrix defined on $\mathbb{R}^{|\mathcal{X}|\times|\mathcal{Y}|}$ and
marginal distributions of $X$ and $Y$ by vectors $P_X$ and $P_Y$ defined on $\mathbb{R}^{|\mathcal{X}|}$ and $\mathbb{R}^{|\mathcal{Y}|}$ given by the row and column sums of $P_{XY}$, respectively. Furthermore, 
we represent the leakage matrix $P_{X|Y}$ by a matrix defined on $\mathbb{R}^{|\mathcal{X}|\times|\mathcal{Y}|}$. Furthermore, for a given $u\in \mathcal{U}$, $P_{X,U}(\cdot,u)$ and $P_{X|U}(\cdot|u)$ defined on $\mathbb{R}^{|\mathcal{X}|}$ are distribution vectors with elements $P_{X,U}(x,u)$ and $P_{X|U}(x|u)$ for all $x\in\cal X$. 
The relation between $U$ and $Y$ is described by the kernel $P_{U|Y}$ defined on $\mathbb{R}^{|\mathcal{U}|\times|\mathcal{Y}|}$. 
We assume that the cardinality of $U$ is fixed and we have $|\mathcal{U}|=K$.
In this work, $P_{X}(x)$, $P_{X}$, $\sqrt{P_{X}}$ and $[P_{X}]$ denote $P_{X}(X=x)$, distribution vector of $X$, a vector with entries $\sqrt{P_X(x)}$, and a diagonal matrix with diagonal entries equal to $P_{X}(x)$, respectively. For two vectors $P$ and $Q$ with same size, we say $P\leq Q$ if $P(x)\leq Q(x)$ for all $x$. 
Furthermore, we assume that the private data is not directly accessible by the agent, resulting in the Markov chain $X - Y - U$.
Our goal is to design the privacy mechanism that produces the disclosed data $U$, which maximizes the utility and satisfies a multi-level point-wise privacy criterion.
In this work, utility is measured by the mutual information $I(U;Y)$, while privacy leakage is quantified by $\mathcal{L}(X \rightarrow u_i)\triangleq D_f(P_{X|U=u_i};P_X)\leq \epsilon_i$, where $u_i \in \mathcal{U}$ for all $i \in \{1,\ldots,K\}$, and $D_f(P;Q)$ denotes a general $f$-divergence between $P$ and $Q$.
 Thus, the privacy problem can be stated as follows 
\begin{subequations}\label{problem}
	\begin{align}
	\sup_{P_{U|Y}} \ \ &I(U;Y),\label{problem1}\\
	\text{subject to:}\ \ &X-Y-U,\label{Markov1}\\
	 &\mathcal{L}(X \rightarrow u_i)\leq \epsilon_i, \forall i\in \{1,\ldots,K\},\label{local1}
	\end{align}
\end{subequations}
Here, we assume that $K$ and the $\epsilon_i$ values are fixed and without loss of generality we have
$
\epsilon_1 \geq \epsilon_2 \geq \cdots \geq \epsilon_K \geq 0 .
$
\begin{remark}
\normalfont
We refer to \eqref{local1} as the \emph{multi-level point-wise privacy constraint}, since each $u_i$ must satisfy a leakage constraint with a different threshold. Furthermore, in this work, we assume that the alphabets of $U$ are known and we design the filter $P_{U|Y}$. This model also covers hybrid cases in which some of the $\epsilon_i$ values are zero. 
Letting $\epsilon_1=\ldots=\epsilon_K$ and $D_f(\cdot||\cdot)=\chi^2(\cdot||\cdot)$ in \eqref{problem}, leads to the problem studied in \cite{khodam}.
Finally, setting $\epsilon_1 = 0$ reduces the model to the perfect privacy problem studied in \cite{borz}.
\end{remark}
\begin{example} (Motivating example)
    Consider a medical dataset in which the private data $X$ represents a patient’s exact diagnosis, the useful data $Y$ contains clinical measurements, and the disclosed data $U$ is a categorical report released to a third party. The alphabet of $U$ may include outcomes such as
\[
u \in \{\text{``healthy''},\ \text{``low risk''},\ \text{``high risk''},\ \text{``critical''}\}.
\]
Revealing $U=\text{``healthy''}$ or $U=\text{``low risk''}$ may be allowed to leak more information about $X$, since these outcomes are less sensitive. In contrast, $U=\text{``high risk''}$ or $U=\text{``critical''}$ must satisfy much stricter leakage constraints due to ethical and legal considerations. As a result, different symbols of $U$ are subject to different privacy leakage thresholds, naturally leading to a multi-level point-wise leakage formulation.
\end{example}
\subsection{Invertible $P_{X|Y}$:} Here, we assume that $|\mathcal{X}|=|\mathcal{Y}|=N$ and $P_{X|Y}\in\mathbb{R}^{N\times N}$ is invertible. Furthermore, we choose the $\chi^2$-distance as the $f$-divergence in \eqref{problem}, and the main problem is formulated as follows.
\begin{align}
    g_{\epsilon_1,\ldots,\epsilon_K}^{\chi}(P_{XY})&\triangleq\sup_{\begin{array}{c} 
	\substack{P_{U|Y}: X-Y-U,\\\chi^2(P_{X|u_i}||P_X)\leq\epsilon_i,\ \forall i\in\{1,\ldots,K\},}
	\end{array}}I(Y;U),\label{main1}
\end{align}
\begin{remark}
\normalfont
    Letting $\epsilon_1=\ldots=\epsilon_K$ in \eqref{main1}, leads to the problem studied in \cite{khodam}. Furthermore, we have
    \begin{align*}
        g_{\epsilon_K,\ldots,\epsilon_K}^{\chi}(P_{XY})&\leq g_{\epsilon_1,\ldots,\epsilon_K}^{\chi}(P_{XY})\\&\leq
        g_{\epsilon_1,\ldots,\epsilon_1}^{\chi}(P_{XY}),
    \end{align*}
    where the upper and lower bounds are studied in \cite{khodam}.
\end{remark}
\subsection{General $P_{X|Y}$:} Here, we first assume that $P_{X|Y}$ has full row rank with $|\mathcal{X}|\leq |\mathcal{Y}|$. Without loss of generality we assume that $P_{X|Y}$ can be represented by two submatrices where the first submatrix is invertible, i.e., $P_{X|Y}=[P_{X|Y_1} , P_{X|Y_2}]$ such that $P_{X|Y_1}$ defined on $\mathbb{R}^{|\mathcal{X}|\times|\mathcal{X}|}$ is invertible. 
 We choose the $\ell_1$-distance as the $f$-divergence in \eqref{problem}, and the main problem is formulated as follows.
\begin{align}
    g_{\epsilon_1,\ldots,\epsilon_K}^{\ell}(P_{XY})&\triangleq\sup_{\begin{array}{c} 
	\substack{P_{U|Y}: X-Y-U,\\\ell_1(P_{X|u_i}||P_X)\leq\epsilon_i,\ \forall i\in\{1,\ldots,K\},}
	\end{array}}I(Y;U),\label{main2}
\end{align}
\begin{remark}\label{r4}
\normalfont
    Letting $\epsilon_1=\ldots=\epsilon_K$ in \eqref{main2}, leads to the problem studied in \cite{Khodam22}. Furthermore, we have
    \begin{align*}
        g_{\epsilon_K,\ldots,\epsilon_K}^{\ell}(P_{XY})&\leq g_{\epsilon_1,\ldots,\epsilon_K}^{\ell}(P_{XY})\\&\leq
        g_{\epsilon_1,\ldots,\epsilon_1}^{\ell}(P_{XY}),
    \end{align*}
    where the upper and lower bounds are studied in \cite{Khodam22}.
\end{remark}
Finally, in this paper, we also discuss how to generalize our framework to an arbitrary leakage matrix.
\section{Main Results}
In this section, we approximate \eqref{main1} and \eqref{main2} and solve the resulting optimization problems. Before stating the main results, we rewrite the conditional distribution $P_{X|U=u_i}$ as a perturbation of $P_X$. Thus, for any $u_i\in\mathcal{U}$, we can write $P_{X|U=u_i}=P_X+\epsilon_{i}\cdot J_{u_i}$, where $J_{u_i}\in\mathbb{R}^{|\mathcal{X}|}$ is a perturbation vector which satisfies following properties.
\begin{align}
\sum_{x\in\mathcal{X}} J_{u_i}(x)&=0,\ \forall u_i,\label{proper1}\\
\sum_{i=1}^{K} \epsilon_{i}P_U(u_i)J_{u_i}(x)&=0,\ \forall x\label{proper2}.
\end{align}
Furthermore, $\forall u_i\in\{u_1,\ldots,u_K\}$, the constraints $\chi^2(P_{X|u_i}||P_X)\leq\epsilon_i$ and $\ell_1(P_{X|u_i}||P_X)\leq\epsilon_i$ can be rewritten as  
\begin{align}
    \|[\sqrt{P_X}^{-1}]J_{u_i}\|_2=\|L_{u_i}\|_2\leq 1,\ \forall u_i\in\{u_1,\ldots,u_K\},
\end{align}
and 
\begin{align}\label{proper3}
    \|J_{u_i}\|_1\leq 1, \ \forall u_i\in\{u_1,\ldots,u_K\},
\end{align}
respectively, where $L_{u_i}\triangleq[\sqrt{P_X}^{-1}]J_{u_i}$ and $u_i$ is the $i$-th alphabet of $U$.
\subsection{Invertible matrix case} 
In the next result, we approximate \eqref{main1} by a quadratic function. 
To do so, let us define $W\triangleq [\sqrt{P_Y}^{-1}]P_{X|Y}^{-1}[\sqrt {P_X}]$.
In all results, the constraint $\sum_i P_U(u_i) = 1$ is used but not written, for brevity of presentation. 
\begin{proposition}\label{prop1}
    For sufficiently small $\epsilon_1$, \eqref{main1} can be approximated by the following problem
    \begin{align}\label{tt}
        \max_{\begin{array}{c} 
		\substack{L_{u_i},P_U: \sum_i \epsilon_{i}P_U(u_i)L_{u_i}=0,\\ \|L_{u_i}\|_2\leq 1,\ L_{u_i}\perp \sqrt{P_X}, \forall i,}
		\end{array}} \!\!\!\!\!\!\!0.5\!\left(\sum_{i=1}^K \!P_U(u_i)\epsilon_{i}^2\|WL_{u_i}\|^2 \right)\!.
    \end{align}
\end{proposition}
\begin{proof}
    The proof is similar to \cite{khodam} and is based on the second order Taylor approximation of the KL-divergence. The only difference is that, due to the multi-level leakage constraints, we have
    \begin{align*}
        P_{X|U=u_i}=P_X+\epsilon_{i}\cdot J_{u_i}, \end{align*}
        which implies
\begin{align*}P_{Y|U=u_i}=P_Y+\epsilon_{i}P_{X|Y}^{-1}\cdot J_{u_i}.
    \end{align*}
    Then, we approximate each term $D(P_{Y|U=u}||P_Y)$ appearing in the mutual information, where $D(\cdot||\cdot)$ denotes the KL-divergence.
\end{proof}
Let $\sigma_{\max}(W)$ and $L^*$ correspond to the maximum singular value and vector of $W$, respectively, where $\|L^*\|_2=1$.
In the next result, we solve \eqref{tt} for binary $U$, i.e., $K=2$. For simplicity of notation, we show $P_U(u_i)$ and $L_{u_i}$ by $P_i$ and $L_i$, respectively.
\begin{proposition}\label{prop2}
    For $K=2$, we have 
    \begin{align}\label{2}
    \eqref{tt}=\frac{1}{2}\epsilon_1\epsilon_2\sigma_{\max}^2(W),
    \end{align}
    attained by
    \begin{align}
        P_{U}(u_1)&=\frac{\epsilon_2}{\epsilon_1+\epsilon_2}, \ P_{U}(u_2)=\frac{\epsilon_1}{\epsilon_1+\epsilon_2},\label{10}\\
        L_{u_1}&=-L_{u_2}=L^*.\label{11}
    \end{align}
\end{proposition}
\begin{proof}
    Clearly, \eqref{10} and \eqref{11} satisfy the constraints in \eqref{tt} leading to a feasible candidate. Furthermore, we have
\begin{align*}
&\frac{1}{2}\left(\sum_{u_i\in\{u_1,u_2\}} \epsilon_i^2P_U(u_i)\|WL_{u_i}\|^2\right)=\\&\frac{1}{2}\sigma_{\text{max}}^2(W)\left( \frac{\epsilon_1^2\epsilon_2}{\epsilon_1+\epsilon_2}+ \frac{\epsilon_2^2\epsilon_1}{\epsilon_1+\epsilon_2}\right)=\frac{1}{2}\epsilon_1\epsilon_2\sigma_{\text{max}}^2(W).
\end{align*}
Thus, \eqref{2} is a lower bound to \eqref{tt}. Next, we show that it is also an upper bound. Using the constraint $P_1\epsilon_1L_1+P_2\epsilon_2L_2=0$, we get $L_2=-\frac{P_1\epsilon_1}{P_2\epsilon_2}L_1$. Substituting $L_2$ in \eqref{tt}, we can rewrite the optimization problem as follows.
\begin{align}
        0.5\!\!\!\!\!\!\max_{\begin{array}{c} 
		\substack{L_1,P_1,P_2: \|L_{1}\|_2\leq \min\{1,\frac{P_2\epsilon_2}{P_1\epsilon_1}\},\\ L_{1}\perp \sqrt{P_X},\ P_1+P_2=1,}
		\end{array}}  \!\!\!\!\!\!\!\!P_1\epsilon_1^2\|WL_{1}\|^2\!+\!\frac{P_1^2\epsilon_1^2}{P_2}\|WL_{1}\|^2\nonumber\\
        =
        0.5\epsilon_1^2\!\!\!\!\!\!\max_{\begin{array}{c} 
		\substack{L_1,P_1,P_2: \|L_{1}\|_2\leq \min\{1,\frac{P_2\epsilon_2}{P_1\epsilon_1}\},\\ L_{1}\perp \sqrt{P_X},\ P_1+P_2=1,}
		\end{array}}  \!\!\!\!\frac{P_1}{P_2}\|WL_{1}\|^2\label{opt2}
    \end{align}
    where the constraint $\|L_{1}\|_2\leq \frac{P_2\epsilon_2}{P_1\epsilon_1}$ follows by $\|L_{2}\|_2\leq 1$ and $L_2=-\frac{P_1\epsilon_1}{P_2\epsilon_2}L_1$. We consider two cases: $\frac{P_2\epsilon_2}{P_1\epsilon_1}\geq 1 \leftrightarrow P_1\leq \frac{\epsilon_2}{\epsilon_1+\epsilon_2}$ or $\frac{P_2\epsilon_2}{P_1\epsilon_1}\leq 1 \leftrightarrow P_2\leq \frac{\epsilon_1}{\epsilon_1+\epsilon_2}$. In the first case, \eqref{opt2} is upper bonded as follows. We have
    \begin{align*}
        \eqref{opt2}&= 0.5\epsilon_1^2\!\!\!\!\!\!\max_{\begin{array}{c} 
		\substack{L_1,P_1: \|L_{1}\|_2\leq 1, L_{1}\perp \sqrt{P_X},}
		\end{array}}  \!\!\!\!\frac{P_1}{1-P_1}\|WL_{1}\|^2\\ &\leq 0.5\epsilon_1^2\sigma_{\max}^2\max_{P_1}\frac{P_1}{1-P_1}= 0.5\epsilon_1\epsilon_2 \sigma_{\max}^2,
    \end{align*}
    where in the last line we used the fact that $\frac{P_1}{1-P_1}$ is an increasing function of $P_1$, hence, we choose $P_1=\frac{\epsilon_2}{\epsilon_1+\epsilon_2}$. In the second case, we have
    \begin{align*}
        \eqref{opt2}&= 0.5\epsilon_1^2\!\!\!\!\!\!\max_{\begin{array}{c} 
		\substack{L_1,P_1,P_2: \|L_{1}\|_2\leq \frac{P_2\epsilon_2}{P_1\epsilon_1}, L_{1}\perp \sqrt{P_X},}
		\end{array}}  \!\!\!\!\frac{P_1}{P_2}\|WL_{1}\|^2\\ &\leq 0.5\epsilon_2^2\sigma_{\max}^2\max_{P_2}\frac{P_2}{1-P_2}= 0.5\epsilon_1\epsilon_2 \sigma_{\max}^2,
    \end{align*}
    where the last line follows by 
    \begin{align*}
        \frac{P_1}{P_2}\|WL_{1}\|^2\leq \frac{P_1}{P_2} \sigma_{\max}^2\left( \frac{P_2\epsilon_2}{P_1\epsilon_1}\right)^2=\frac{P_2}{P_1}\sigma_{\max}^2(\frac{\epsilon_2}{\epsilon_1})^2.
    \end{align*}
    In both cases, the upper bound is attained by \eqref{10} and \eqref{11}. Hence, for $K=2$, we have $\eqref{tt}=\frac{1}{2}\epsilon_1\epsilon_2 \sigma_{\max}^2$.
\end{proof}
Next, we show that, without loss of optimality in \eqref{tt}, we can assume that $U$ is binary and has nonzero leakage at the two highest leakage thresholds, namely $\epsilon_1$ and $\epsilon_2$. Similarly as Proposition \ref{prop2}, for simplicity of notation, we show $P_U(u_i)$ and $L_{u_i}$ by $P_i$ and $L_i$.
\begin{lemma}\label{lem1}
To solve \eqref{tt}, without loss of optimality, we can assume that $K=2$, and for other letters of $U$ we have
\begin{align*}
    P_U(u_i)=0, \ i\geq 3.
\end{align*}
In other words, we have
\begin{align}\label{12}
\eqref{tt}= 0.5\!\!\!\!\!\!\!\!\!\!\!\max_{\begin{array}{c} 
		\substack{L_1,L_2,P_1,P_2: \|L_{1}\|_2,\|L_{2}\|_2\leq 1,\\ L_1,L_{2}\perp \sqrt{P_X},\ P_1\epsilon_1L_1+P_2\epsilon_2L_2=0,}
		\end{array}}  \!\!\!\!\!\!\!\!\!\!\!\!\!\!\!P_1\epsilon_1^2\|WL_{1}\|^2\!+\!P_2\epsilon_2^2\|WL_{2}\|^2,
\end{align}
where $P_1+ P_2=1$.
\end{lemma}
\begin{proof}
    The proof is similar to \cite[Proposition 4.]{khodam}. Let $\{L_u^*,P_U^*\}$ be the maximizer of \eqref{tt} and $\epsilon_{1}\geq \epsilon_{2}\geq \ldots \geq \epsilon_{K}$. Furthermore, let $L_1^*$ and $L_2^*$ attain the highest and second highest values for $\|WL_{u_i}^*\|^2$, respectively, and let $u_m$ and $u_n$ be the corresponding letters. Then, we have
    \begin{align*}
        &\sum_i \epsilon_i^2P_U^*(u_i)\|WL_{u_i}^*\|^2\\&\stackrel{(a)}{\leq}\left(\sum_{i} \epsilon_{i}^2P_U^*(u_i)\right)\|W L_1^*\|^2+\left(\sum_{j} \epsilon_{j}^2P_U^*(u_j)\right)\|W L_2^*\|^2\\&\stackrel{(b)}{\leq}\epsilon_{1}^2\left(\sum_{i} P_U^*(u_i)\right)\|W L_1^*\|^2+\epsilon_{2}^2\left(\sum_{j} P_U^*(u_j)\right)\|W L_2^*\|^2 \\&\stackrel{(c)}{=} \epsilon_{1}^2P_1\|W L_1^*\|^2+\epsilon_{2}^2P_2\|W L_2^*\|^2,
    \end{align*}
    where in step (a) we divide the letters of $U^*$ into two sets where the first includes $u_{m}$ and the second includes $u_n$. The remaining letters are chosen randomly. In step (b) we use $\epsilon_{1}\geq \epsilon_{2}\geq \ldots \geq \epsilon_{K}$. Finally, in step (c), we let $\sum_{i} P_U^*(u_i)=P_1$ and $\sum_{j} P_U^*(u_j)=P_2$, where $P_1+P_2=1$. We recall that $\sum_{i} P_U^*(u_i)=P_1$ corresponds to the sum of the weights of the first set, which includes $u_m$, and that $\sum_{j} P_U^*(u_j)=P_2$ denotes the sum of the weights of the second set, which includes $u_n$. In other words, $m$ is included in the first summation, while $n$ is included in the second.
    Since the final upper bound can be attained by \eqref{10} and \eqref{11}, \eqref{12} is proved.
\end{proof}
\begin{corollary}\label{cor1}
    If $\epsilon_3<\epsilon_2$, then for the maximizer of \eqref{tt}, we have
    \begin{align*}
        P_U(u_3)=\ldots=P_{U}(u_K)=0.
    \end{align*}
    In other words, $U$ should be a binary random variable that leaks information about $X$ only through the two letters with the highest privacy budgets.
\end{corollary}
\begin{proof}
    The proof follows by Lemma \ref{lem1}. Let $\epsilon_3<\epsilon_2$. In this case, if there exists $u_i$ with $P_U(u_i)>0$ for some $i\geq 3$, then (b) becomes a strict inequality. Furthermore, since the final upper bound is attainable, the utility achieved by any $U$ with $P_U(u_i) > 0$ for some $i \geq 3$ is strictly less than the final upper bound achieved by \eqref{10} and \eqref{11}. This completes the proof. 
\end{proof}
By using propositions \ref{prop1} and \ref{prop2}, and Lemma \ref{lem1}, we obtain the following theorem. In the next result, the approximation notation $\simeq$ corresponds to a second-order Taylor expansion. In other words, $f(\epsilon) \simeq g(\epsilon)$ if
$
f(\epsilon) = g(\epsilon) + o(\epsilon^2),
$
where $o(\cdot)$ denotes the Bachmann--Landau notation.
\begin{theorem}
For sufficiently small $\epsilon_1$, we have
\begin{align}
    g_{\epsilon_1,\ldots,\epsilon_K}^{\chi}(P_{XY}) \simeq \eqref{tt} = \frac{1}{2}\epsilon_1\epsilon_2\sigma_{\max}^2(W).
\end{align}
The equality is attained by a binary $U$ with \eqref{10} and \eqref{11}. Furthermore, when $\epsilon_2>\epsilon_3$, the disclosed data which attains the maximum of \eqref{tt} is a binary RV that leaks information about $X$ only through the two letters with the highest privacy budgets.
\end{theorem}
\begin{proof}
    The proof follows from Propositions \ref{prop1} and \ref{prop2}, Lemma \ref{lem1} and Corollary \ref{cor1}.
 \end{proof}
 After finding the optimizer of \eqref{tt}, the final privacy mechanism can be obtained as
 \begin{align}\label{filt}
P_{Y|U=u_1}&=P_X+\epsilon_1 P_{X|Y}^{-1}[\sqrt{P_X}]L_{u_1},\\
P_{Y|U=u_2}&=P_X+\epsilon_2 P_{X|Y}^{-1}[\sqrt{P_X}]L_{u_2}.
\end{align}
Here, $L_{u_1}$ and $L_{u_2}$ are given in \eqref{11}. Furthermore, the marginal distribution of $U$ is given by \eqref{10}.
\begin{remark}
    Using \cite{khodam}, we obtain \begin{align*}g_{\epsilon_1,\ldots,\epsilon_1}^{\chi}(P_{X,Y})&\simeq \frac{1}{2}\epsilon_1^2\sigma_{\max}^2(W)
    \\&\geq \frac{1}{2}\epsilon_1\epsilon_2\sigma_{\max}^2(W) \simeq g_{\epsilon_1,\ldots,\epsilon_K}^{\chi}(P_{XY}). 
    \end{align*}
    Furthermore, \begin{align*}g_{\epsilon_1,\ldots,\epsilon_K}^{\chi}(P_{X,Y})&\simeq \frac{1}{2}\epsilon_1\epsilon_2\sigma_{\max}^2(W)
    \\&\geq \frac{1}{2}\epsilon_K^2\sigma_{\max}^2(W) \simeq g_{\epsilon_K,\ldots,\epsilon_K}^{\chi}(P_{XY}). 
    \end{align*}
\end{remark}
\begin{corollary}
    The optimizer of \eqref{tt}, given by \eqref{10} and \eqref{11}, leads to a lower bound on
$g_{\epsilon_1,\ldots,\epsilon_K}^{\chi}(P_{XY})$ as follows:
\begin{align}
g_{\epsilon_1,\ldots,\epsilon_K}^{\chi}(P_{XY}) \geq I(U^*;Y).
\end{align}
Here, $I(U^*;Y)$ is evaluated using the privacy mechanism in \eqref{filt} and the marginal
distribution of $U^*$ given by \eqref{10}.
\end{corollary}
\begin{proof}
    The proof follows from the fact that the privacy mechanism given by \eqref{filt} and \eqref{10} satisfies the constraints in \eqref{main1}.
\end{proof}
\subsection{General $P_{X|Y}$}
In this section, we approximate \eqref{main2} and solve the resulting optimization problem. To do so, we first consider a full row rank matrix $P_{X|Y}$. We then discuss how to generalize it for any arbitrary matrix.

Similar to \cite{Khodam22}, we show that for any feasible $U$ which satisfies the constraints in \eqref{main2}, $P_{Y|U=u_i}$ lies in a convex polytope.
To do so, let $M\in \mathbb{R}^{|\mathcal{X}|\times|\mathcal{Y}|}$ be constructed as follows:
 Let $V$ be the matrix of right eigenvectors of $P_{X|Y}$, i.e., $P_{X|Y}=U\Sigma V^T$ and $V=[v_1,\ v_2,\ ... ,\ v_{|\mathcal{Y}|}]$, then $M$ is defined as
 \begin{align}\label{M}
 M \triangleq \left[v_1,\ v_2,\ ... ,\ v_{|\mathcal{X}|}\right]^T. 
 \end{align}  
\begin{lemma}\label{lem4}
 	 Let $J_{u_i}$ satisfy the three properties \eqref{proper1}, \eqref{proper2} and \eqref{proper3}. For sufficiently small $\epsilon>0$, for every $u_i\in\mathcal{U}$, the vector $P_{Y|U=u_i}$ belongs to the following convex polytope
 	\begin{align*}
 	\mathbb{S}_{u_i} = \left\{y\in\mathbb{R}^{|\mathcal{Y}|}|My=MP_Y+\epsilon_i M\begin{bmatrix}
 	P_{X|Y_1}^{-1}J_{u_i}\\0
 	\end{bmatrix},\ y\geq 0\right\},
 	\end{align*}	
 	where $\begin{bmatrix}
 	P_{X|Y_1}^{-1}J_{u_i}\\0
 	\end{bmatrix}\in\mathbb{R}^{|\cal Y|}$.
 \end{lemma}
\begin{proof}
The proof is similar to \cite[Lemma 2]{Khodam22} and is based on properties of null space of $M$ given in \cite[Lemma 1]{Khodam22}, and Markov chain $X-Y-U$. The only difference is that, for each vector $P_{Y \mid U = u_i}$, we use different leakage as $\epsilon_i$. Furthermore, $J_{u_i}$ satisfies \eqref{proper2}, which includes an additional term $\epsilon_i$. This does not affect the proof.
\end{proof}
Next, similar to \cite[Theorem 1]{Khodam22}, we have the following equivalency. Instead of maximizing $I(Y;U)$, we minimize $H(Y|U)$.
\begin{proposition}
    Let $J_{u_i}$ satisfy the three properties \eqref{proper1}, \eqref{proper2} and \eqref{proper3}. We have
    \begin{align}\label{equi}
	\min_{\begin{array}{c} 
		\substack{P_{U|Y}:X-Y-U\\ \ \ell_1(P_{X|u_i}||P_X)\leq \epsilon_i,\ \forall u_i,}
		\end{array}}\! \! \! \!\!\!\!\!\!\!\!\!\!\!\!\!\!\!\!H(Y|U) =\!\!\!\!\!\!\!\!\! \min_{\begin{array}{c} 
		\substack{P_U,\ P_{Y|U=u_i}\in\mathbb{S}_{u_i},\ \forall u_i\in\mathcal{U},\\ \sum_i P_U(u_i)P_{Y|U=u_i}=P_Y,\\ J_{u_i} \text{satisfies}\ \eqref{proper1},\ \eqref{proper2},\ \text{and}\ \eqref{proper3}}
		\end{array}} \!\!\!\!H(Y|U).
	\end{align}
\end{proposition}
\begin{proof}
    The proof follows arguments similar to those in \cite[Theorem 1]{Khodam22}.
\end{proof}
Next, we discuss how $H(Y|U)$ is minimized over $P_{Y|U=u_i}\in\mathbb{S}_{u_i}$ for all $u_i\in\mathcal{U}$.
\begin{proposition}\label{4}
	Let $P^*_{Y|U=u_i},\ \forall u_i\in\mathcal{U}$ be the minimizer of $H(Y|U)$ over the set $\{P_{Y|U=u_i}\in\mathbb{S}_{u_i},\ \forall u_i\in\mathcal{U}|\sum_i P_U(u_i)P_{Y|U=u_i}=P_Y\}$, then 
	$P^*_{Y|U=u_i}\in\mathbb{S}_{u_i}$ for all $u_i\in\mathcal{U}$ must belong to the extreme points of $\mathbb{S}_{u_i}$.
\end{proposition}
\begin{proof}
    The proof is similar to \cite[Proposition 3]{Khodam22} and is based on the concavity of entropy.
Intuitively, if $P_{Y | U = u_i}$ is not an extreme point of $\mathbb{S}_{u_i}$, it can be written
as a convex combination of extreme points. Then, it can be shown that the entropy
$H(P_{Y | U = u_i})$ can be reduced. Hence, the minimizers lie among the extreme points of $\mathbb{S}_{u_i}$ for each $u_i$.
\end{proof}
Following the same approach in \cite{Khodam22} we can find the extreme points and approximate the right hand side in \eqref{equi}.
To find the extreme points of $\mathbb{S}_{u_i}$ let $\Omega_{u_i}$ be the set of indices which correspond to $|\mathcal{X}|$ linearly independent columns of $M$, i.e., $|\Omega_{u_i}|=|\mathcal{X}|$ and $\Omega_{u_i}\subset \{1,..,|\mathcal{Y}|\}$. Let $M_{\Omega_{u_i}}\in\mathbb{R}^{|\mathcal{X}|\times|\mathcal{X}|}$ be the submatrix of $M$ with columns indexed by the set $\Omega_{u_i}$. Assume that $\Omega_{u_i} = \{\omega_1,..,\omega_{|\mathcal{X}|}\}$, where $\omega_j\in\{1,..,|\mathcal{Y}|\}$ and all elements are arranged in an increasing order. The $\omega_j$-th element of the extreme point $V_{\Omega_{u_i}}^*$ can be found as $j$-th element of $M_{\Omega_{u_i}}^{-1}(MP_Y+\epsilon_i M\begin{bmatrix}
P_{X|Y_1}^{-1}J_{u_i}\\0\end{bmatrix})$, i.e., for $1\leq j \leq |\mathcal{X}|$ we have
\begin{align}\label{defin1}
V_{\Omega_{u_i}}^*(\omega_j)= \left(M_{\Omega_{u_i}}^{-1}MP_Y+\epsilon_i M_{\Omega_{u_i}}^{-1}M\begin{bmatrix}
P_{X|Y_1}^{-1}J_{u_i}\\0\end{bmatrix}\right)(j).
\end{align}
Other elements of $V_{\Omega_{u_i}}^*$ are set to be zero. Now we approximate the entropy of $V_{\Omega_{u_i}}^*$.\\
	Let $V_{\Omega_{u_i}}^*$ be an extreme point of the set $\mathbb{S}_u$, then we have
	\begin{align}
	H(P_{Y|U=u_i}) &=\sum_{y=1}^{|\mathcal{Y}|}-P_{Y|U=u}(y)\log(P_{Y|U=u}(y))\nonumber\\&=-(b_{u_i}+\epsilon_i a_{u_i}J_{u_i})+o(\epsilon_i),\label{koonkos}
	\end{align}
	with $$b_{u_i} = l_{u_i} \left(M_{\Omega_{u_i}}^{-1}MP_Y\right),$$ 
	$$a_{u_i} = l_{u_i}\left(M_{\Omega_{u_i}}^{-1}M(1:|\mathcal{X}|)P_{X|Y_1}^{-1}\right)\in\mathbb{R}^{1\times|\mathcal{X}|},$$
	$$l_{u_i} = \left[\log\left(M_{\Omega_u}^{-1}MP_{Y}(j)\right)\right]_{j=1:|\mathcal{X}|}\in\mathbb{R}^{1\times|\mathcal{X}|},
	$$ and $M_{\Omega_{u_i}}^{-1}MP_{Y}(j)$ stands for $j$-th ($1\leq j\leq |\mathcal{X}|$) element of the vector $M_{\Omega_{u_i}}^{-1}MP_{Y}$. Furthermore, $M(1:|\mathcal{X}|)$ stands for submatrix of $M$ with first $|\mathcal{X}|$ columns.
The proof of \eqref{koonkos} follows similar lines as \cite[Lemma~4]{Khodam22} and is based on first order Taylor expansion of $\log(1+x)$.
By using \eqref{koonkos} we can approximate \eqref{equi} as follows.\\
\begin{theorem}\label{baghal}
	For sufficiently small $\epsilon_1$, the minimization problem in \eqref{equi} can be approximated as follows
	\begin{align}\label{kospa}
	&\min_{P_U(.),\{J_{u_i}, u_i\in\mathcal{U}\}} -\left(\sum_{i=1}^{K} P_U(u_i)(b_{u_i}+\epsilon_{i} a_{u_i}J_{u_i})\right),\\\nonumber
	&\text{subject to:}\\\nonumber
	&\sum_{i=1}^{K} P_U(u_i)V_{\Omega_{u_i}}^*=P_Y,\ \sum_{i=1}^{K} P_U(u_i)\epsilon_{i} J_{u_i}=0,\ P_U\in \mathbb{R}_{+}^{K},\\\nonumber
	&\|J_{u_i}\|_1\leq 1,\  \bm{1}^T\cdot J_{u_i}=0,\ \forall u_i\in\mathcal{U}.\nonumber
	\end{align} 
    Furthermore, \eqref{kospa} can be rewritten as a linear program.
    \end{theorem}
\begin{proof}
    The proof of \eqref{kospa} follows directly from \eqref{koonkos} and the fact that the minimum of $H(Y|U)$ occurs at the extreme points of the sets $\mathbb{S}_{u_i}$.
    To prove the second statement of this theorem, we note that
    by using the vector
$$\eta_{u_i}=P_U(u_i)\left(M_{\Omega_{u_i}}^{-1}MP_Y+\epsilon_{i} M_{\Omega_{u_i}}^{-1}M(1:|\mathcal{X}|)P_{X|Y_1}^{-1}J_{u_i}\right)$$ for all $u_i\in \mathcal{U}$, where $\eta_{u_i}\in\mathbb{R}^{|\mathcal{X}|}$, we can write \eqref{kospa} as a linear program. The vector $\eta_{u_i}$ corresponds to a multiple of non-zero elements of the extreme point $V_{\Omega_{u_i}}^*$. Furthermore, $P_U(u_i)$ and $J_{u_i}$ can be uniquely found as
\begin{align*}
P_U(u_i)&=\bm{1}^T\cdot \eta_{u_i},\\
J_{u_i}&=\frac{P_{X|Y_1}M(1:|\mathcal{X}|)^{-1}M_{\Omega_{u_i}}[\eta_{u_i}\!-\!(\bm{1}^T \eta_{u_i})M_{\Omega_{u_i}}^{-1}MP_Y]}{\epsilon_i(\bm{1}^T\cdot \eta_{u_i})}.
\end{align*}
Next, we find the linear program using \eqref{kospa} and $\eta_{u_i}$. For the cost function, we have
\begin{align*}
&-\left(\sum_{i=1}^{K} P_U(u_i)b_{u_i}+\epsilon_i P_U(u_i)a_{u_i}J_{u_i}\right)=
-\sum_i b_i(\bm{1}^T\eta_{u_i})\\&-\epsilon_i\sum_i \!\!a_{u_i}\! \left[P_{X|Y_1}M(1:|\mathcal{X}|)^{-1}M_{\Omega_{u_i}}\![\eta_{u_i}\!-\!(\bm{1}^T \eta_{u_i})M_{\Omega_{u_i}}^{-1}MP_Y]\right]\!,
\end{align*}
which is a linear function of elements of $\eta_{u_i}$ for all $u_i\in\mathcal{U}$.
Non-zero elements of the vector $P_U(u_i)V_{\Omega_{u_i}}^*$ equal to the elements of $\eta_{u_i}$, i.e., we have
$$
P_U(u_i)V_{\Omega_{u_i}}^*(\omega_j)=\eta_{u_i}(j).
$$
Thus, the constraint $\sum_{i=1}^{K} P_U(u_i)V_{\Omega_{u_i}}^*=P_Y$ can be rewritten as a linear function of elements of $\eta_{u_i}$. For the constraints $\sum_{i=1}^{K} \epsilon_iP_U(u_i)J_{u_i}=0$, $P_U(u_i)\geq 0,\forall u_i,$ and $\sum_{j=1}^{|\mathcal{X}|}J_{u_i}(x_j)=0$, we have
\begin{align*}
&\sum_iP_{X|Y_1}M(1:|\mathcal{X}|)^{-1}M_{\Omega_{u_i}}\left[\eta_{u_i}-(\bm{1}^T\cdot\eta_{u_i})M_{\Omega_{u_i}}^{-1}MP_Y\right]=0,\\
& \sum_j \eta_{u_i}(j)\geq 0,\ \forall u_i\in\mathcal{U},\\
&\bm{1}^T P_{X|Y_1}M(1:|\mathcal{X}|)^{-1}M_{\Omega_{u_i}}\left[\eta_{u_i}-(\bm{1}^T\cdot\eta_{u_i})M_{\Omega_{u_i}}^{-1}MP_Y\right]=0.
\end{align*}
Furthermore, for the last constraint $\sum_j |J_{u_i}(x_j)|\leq 1$ we have 
\begin{align*}
&\sum_j\! \left|\left(P_{X|Y_1}M(1\!:\!|\mathcal{X}|)^{-1}M_{\Omega_{u_i}}\!\left[\eta_{u_i}\!-\!(\bm{1}^T\!\eta_{u_i})M_{\Omega_{u_i}}^{-1}MP_Y\right]\right)(j)\right|\! \leq\\&\epsilon_{u_i}(\bm{1}^T\eta_{u_i}),\ \forall u_i\in\mathcal{U}.
\end{align*}
The last constraint includes absolute values that can be handled by considering two cases for each absolute value. Thus, all constraints can be rewritten as a linear function of elements of $\eta_{u_i}$ for all $u_i$.
\end{proof}
\section{Discussions}
In this section, we first discuss the hybrid case in which some letters of $U$
satisfy perfect privacy. We then discuss how the framework can be extended to an
arbitrary matrix $P_{X | Y}$.
\subsection{Hybrid case}
In this part, we discuss cases in which some of the privacy budgets are zero, that is, perfect privacy for some letters. To this end, let $\epsilon_1 \geq \cdots \geq \epsilon_t \geq \epsilon_{t+1} = \cdots = \epsilon_K = 0$,
that is, the letters $\{u_{t+1}, \ldots, u_K\}$ satisfy perfect privacy. In this case, for all $u\in\{u_{t+1},\ldots,u_K\}$ the convex polytope in Lemma 2 is given by
\begin{align}
 	\mathbb{S}_{u_i} = \left\{y\in\mathbb{R}^{|\mathcal{Y}|}|My=MP_Y,\ y\geq 0\right\}.
\end{align}
Furthermore, for $u_i\in\{u_{t+1},\ldots,u_K\}$, the extreme points in \eqref{defin1} become
\begin{align*}
V_{\Omega_{u_i}}^*(\omega_j)= \left(M_{\Omega}^{-1}MP_Y\right)(j).
\end{align*}
We note that, for such letters, we do not approximate the entropy $H(P_{Y|U=u_i})$, and we obtain the following result.
\begin{theorem}
Let $\epsilon_1 \geq \cdots \geq \epsilon_t \geq \epsilon_{t+1} = \cdots = \epsilon_K = 0$. For sufficiently small $\epsilon_1$, the minimization problem in \eqref{equi} can be approximated as follows
\begin{align}\label{kospa2}
	&\min_{P_U(.),\{J_{u_i}, u_i\in\mathcal{U}_1\}} -\left(\sum_{i=1}^{K} P_U(u_i)b_{u_i}+\sum_{i=1}^t P_U(u_i)\epsilon_i a_{u_i}J_{u_i}\right)\\\nonumber
	&\text{subject to:}\\\nonumber
	&\sum_{i=1}^{K} P_U(u_i)V_{\Omega_{u_i}}^*=P_Y,\ \sum_{i=1}^{t} P_U(u_i)\epsilon_i J_{u_i}=0,\ P_U\in \mathbb{R}_{+}^{K},\\\nonumber
	&\|J_{u_i}\|_1\leq 1,\  \bm{1}^T\cdot J_{u_i}=0,\ \forall u_i\in\mathcal{U}_1.\nonumber
	\end{align} 
    where $\mathcal{U}_1=\{u_1,\ldots,u_t\}$.
    Furthermore, \eqref{kospa2} can be rewritten as a linear program.
\end{theorem}
\begin{proof}
    The proof is similar to that of Theorem~\ref{baghal}. The only difference is that, for
$u_i \in \{u_{t+1}, \ldots, u_K\}$, we do not approximate the entropy
$H(Y | U = u_i)$. Furthermore, for these letters, the entropy equals $b_{u_i}$. For the linear program, we use the same transformation as in Theorem~\ref{baghal}
for $u_i \in \mathcal{U}_1$. For the remaining realizations of $U$, we use
\[
\eta_{u_i} = P_U(u_i)\left(M_{\Omega_{u_i}}^{-1} M P_Y\right).
\]
\end{proof}
\begin{remark}
    In contrast to the invertible case, in the hybrid case the letters with zero
leakage can still attain nonzero weights in the utility term, that is,
$P_U(u_i) > 0$ even when $P_{X \mid U = u_i} = P_X$. For more details, see the
numerical example.
\end{remark}
\subsection{Generalization to arbitrary $P_{X|Y}$}
Here, we discuss how the framework for approximating \eqref{main2} can be
generalized to any $P_{X | Y}$. By checking the proof of \cite[Lemma 2]{Khodam22}, we observe that the convex
polytope $\mathbb{S}_{u_i}$ arises from the following equation:
\begin{align}\label{24}
P_{X | U = u_i} - P_X
= P_{X | Y}\bigl[P_{Y | U = u_i} - P_Y\bigr]
= \epsilon_i J_{u_i}.
\end{align}
When $P_{X \mid Y}$ has full row rank, we can use its invertible submatrix
$P_{X \mid Y_1}$ to obtain a particular solution, and use the remaining part
to describe the component that lies in the null space of $P_{X \mid Y}$. However, for a general $P_{X|Y}$, we can use pseudo inverse of $P_{X|Y}$ for the particular solution instead. In more details, using \eqref{24} we have
\begin{align}\label{25}
    P_{Y|U=u_i}=P_{Y}+\epsilon_i P_{X|Y}^{\dagger}J_{u_i}+z^*,
\end{align}
where $P_{X|Y}^{\dagger}$ is the pseudo inverse of $P_{X|Y}$ and $z^*\in \text{null}(P_{X|Y})$. 
The main difference from the full row-rank case is that here the term
$\epsilon_i P_{X \mid Y}^{\dagger} J_{u_i}$ serves as the particular solution,
instead of
\[
\epsilon_i
\begin{bmatrix}
P_{X \mid Y_1}^{-1} J_{u_i} \\
0
\end{bmatrix}.
\]
Then by using properties of matrix $M$ stated in \cite[Lemma 1]{Khodam22}, we have the following result.
\begin{lemma}\label{lem42}
 	 Let $J_{u_i}$ satisfy the three properties \eqref{proper1}, \eqref{proper2} and \eqref{proper3}. For sufficiently small $\epsilon_1>0$, for every $u_i\in\mathcal{U}$, the vector $P_{Y|U=u_i}$ belongs to the following convex polytope
 	\begin{align*}
 	\mathbb{S}_{u_i} = \left\{y\in\mathbb{R}^{|\mathcal{Y}|}|My=MP_Y+\epsilon_iMP_{X|Y}^{\dagger}J_{u_i},\ y\geq 0\right\}.
 	\end{align*}	
 \end{lemma}
Following the same procedure as discussed before the extreme points of $\mathbb{S}_{u_i}$ are given as
\begin{align}\label{defin11}
V_{\Omega_{u_i}}^*(\omega_j)= \left(M_{\Omega_{u_i}}^{-1}MP_Y+\epsilon_i M_{\Omega_{u_i}}^{-1}MP_{X|Y}^{\dagger}J_{u_i}\right)(j).
\end{align}
Furthermore, the entropy of $V_{\Omega}^*(\omega_i)$ can be approximated as follows 
\begin{align}
	H(P_{Y|U=u}) \simeq-(b_{u_i}+\epsilon_i a_{u_i}'J_{u_i}),\label{koonkos2}
	\end{align}
	with $$b_{u_i} = l_{u_i} \left(M_{\Omega_{u_i}}^{-1}MP_Y\right),$$ 
	$$a_{u_i}' = l_{u_i}\left(M_{\Omega_{u_i}}^{-1}MP_{X|Y}^{\dagger}\right)\in\mathbb{R}^{1\times|\mathcal{X}|},$$
	$$l_{u_i} = \left[\log\left(M_{\Omega_{u_i}}^{-1}MP_{Y}(j)\right)\right]_{j=1:|\mathcal{X}|}\in\mathbb{R}^{1\times|\mathcal{X}|}.$$
Note that, in the proof of approximating the extreme points in
\cite[Lemma 4]{Khodam22}, we use the the fact that
\[
1^{T} M_{\Omega_{u_i}}^{-1} M(1\!:\!|\mathcal{X}|) P_{X \mid Y_1}^{-1} = 0 .
\]
Here, the extension to the pseudo inverse also holds. Specifically, we need to
show that
\[
1^{T} M_{\Omega_{u_i}}^{-1} M P_{X \mid Y}^{\dagger} = 0 .
\]
Using the same arguments as in \cite[Lemma 4]{Khodam22}, it suffices to show that
\[
1^{T} P_{X \mid Y}^{\dagger} = 1^{T} .
\]
This follows from
\[
1^{T} = 1^{T} I = 1^{T} P_{X \mid Y} P_{X \mid Y}^{\dagger}
= 1^{T} P_{X \mid Y}^{\dagger},
\]
where we used $P_{X \mid Y} P_{X \mid Y}^{\dagger} = I$ and
$1^{T} P_{X \mid Y} = 1^{T}$. For further details, see
\cite[Lemma 4]{Khodam22}.
Hence, we extend Theorem \ref{baghal} as follows.
\begin{theorem}\label{baghal5}
	For any leakage matrix $P_{X|Y}$ and for sufficiently small $\epsilon_1$, the minimization problem in \eqref{equi} can be approximated as follows
	\begin{align}\label{kospa22}
	&\min_{P_U(.),\{J_{u_i}, u_i\in\mathcal{U}\}} -\left(\sum_{i=1}^{K} P_U(u_i)(b_{u_i}+\epsilon_i a_{u_i}'J_{u_i})\right)\\\nonumber
	&\text{subject to:}\\\nonumber
	&\sum_{i=1}^{K} P_U(u_i)V_{\Omega_{u_i}}^*=P_Y,\ \sum_{i=1}^{K} P_U(u_i)\epsilon_i J_{u_i}=0,\ P_U\in \mathbb{R}_{+}^{K},\\\nonumber
	&\|J_{u_i}\|_1\leq 1,\  \bm{1}^T\cdot J_{u_i}=0,\ \forall u_i\in\mathcal{U}.\nonumber
	\end{align} 
    Furthermore, \eqref{kospa22} can be rewritten as a linear program.
    \end{theorem}
\begin{proof} The proof follows similar arguments as those in Theorem 3. The only difference is that we use \eqref{defin11} instead of \eqref{defin1}, and we use $a_{u_i}'$ instead of $a_{u_i}$.\end{proof} 
\begin{remark}
We emphasize that, in \eqref{25}, we use the general solution to $Ax = b$.
However, there are cases in which this system has no solution. For instance,
if $b$ does not lie in the column space of $A$, then there exists no $x$
satisfying $Ax = b$. In such cases, $A^{\dagger} b$ minimizes $\|Ax - b\|$.

To address this issue, we follow the procedure described above and solve the
corresponding linear program. After obtaining $J_{u_i}$, we check whether
\eqref{24} is satisfied. If \eqref{24} is not satisfied, this indicates that
the feasible set is empty and that nonzero utility cannot be achieved.

This situation can occur when $P_{X \mid Y}$ is tall, that is, when
$|\mathcal{X}| > |\mathcal{Y}|$. In practice, however, we typically have
$|\mathcal{X}| < |\mathcal{Y}|$. The framework developed in this paper applies
to cases in which the feasible set is non-empty. In particular, we ignore cases
for which $\eqref{main2} = 0$.
\end{remark}
\subsection{Extending the permissible leakage range}
In this section, we discuss how our method can be applied to larger permissible leakage intervals. In this paper, to approximate $I(U;Y)$, we use a first-order Taylor expansion of $\log(1+x)$ at the extreme points of the convex polytopes $\mathbb{S}_{u_i}$. One way to increase the permissible leakage interval is to employ higher-order approximations of $\log(1+x)$. This follows because higher-order Taylor expansions allow the same approximation error to be achieved over a larger range of $x$. Specifically, if the first-order approximation is valid over the interval $x \leq c_1$, then a higher-order approximation can be valid over an interval $x \leq c_2$, where $c_2 > c_1$, while maintaining the same approximation error. For instance, let $$\log(1+x)=x-\frac{x^2}{2}+\text{error}_1(x),$$ and $$\log(1+x)=x-\frac{x^2}{2}+\frac{x^3}{3}+\text{error}_2(x).$$ If for the first approximation we use the interval $x\leq c_1$, then we can use the interval $x\leq c_2$ where $c_2>c_1$, for the second approximation to have the same error, i.e., $$|\text{error}_1(x)|=|\text{error}_2(x)|.$$ Furthermore, we obtain tighter bounds compared to the current results (closer to the optimal value); however, the optimization problems become more complex.
\section{Numerical Example}
Here, we use \cite[Example 2]{Khodam22} and compare our design under multi-level privacy constraints with it. We study a hybrid case as follows.
Let $P_{X|Y}=\begin{bmatrix}
	0.3  \ 0.8 \ 0.5 \ 0.4\\0.7 \ 0.2 \ 0.5 \ 0.6
	\end{bmatrix}$  
	and $P_Y=[\frac{1}{2},\ \frac{1}{4},\ \frac{1}{8},\ \frac{1}{8}]^T$.
    Furthermore, let $\epsilon_1=\epsilon_2=\epsilon_3=0.01$ and $\epsilon_4=0$. 
    Let $J_{u_i}=\begin{bmatrix}
J_{u_i}^1\\J_{u_i}^2
\end{bmatrix}$ for $u_i\in\{u_1,u_2,u_3\}$. Using \eqref{proper1} we obtain $J_{u_i}^1+J_{u_i}^2=0$. Thus, we show $J_{u_i}$ by $\begin{bmatrix}
-J_{u_i}^2\\J_{u_i}^2
\end{bmatrix}$.
    For $u_i\in\{u_1,u_2,u_3\}$, the extreme points of $\mathbb{S}_{u_i}$ are as follows
    \begin{align*}
	&V_{1_{\Omega_{u_i}}}^* = \begin{bmatrix}
	0.675-2\epsilon_iJ_{u_i}^2\\0.325+2\epsilon_iJ_{u_i}^2\\0\\0
	\end{bmatrix},\
	V_{2_{\Omega_{u_i}}}^* = \begin{bmatrix}
	0.1875+5\epsilon_iJ_{u_i}^2\\0\\0.8125-5\epsilon_iJ_{u_i}^2\\0
	\end{bmatrix}\\
	& V_{3_{\Omega_{u_i}}}^* = \begin{bmatrix}
	0\\0.1563- 2.5\epsilon_iJ_{u_i}^2\\0\\0.8437+ 2.5\epsilon_iJ_{u_i}^2
	\end{bmatrix},\
	V_{4_{\Omega_{u_i}}}^* = \begin{bmatrix}
	0\\0\\0.6251- 10\epsilon_iJ_{u_i}^2\\0.3749+ 10\epsilon_iJ_{u_i}^2
	\end{bmatrix},
	\end{align*}
    where $\epsilon_i=0.01$. The set of extreme points for $u_i=u_4$ is as follows
    \begin{align*}
	&V_{1_{\Omega_{u_4}}}^* = \begin{bmatrix}
	0.675\\0.325\\0\\0
	\end{bmatrix},\
	V_{2_{\Omega_{u_4}}}^* = \begin{bmatrix}
	0.1875\\0\\0.8125\\0
	\end{bmatrix}\\
	& V_{3_{\Omega_{u_4}}}^* = \begin{bmatrix}
	0\\0.1563\\0\\0.8437
	\end{bmatrix},\
	V_{4_{\Omega_{u_4}}}^* = \begin{bmatrix}
	0\\0\\0.6251\\0.3749
	\end{bmatrix}.
	\end{align*}
    After examining all possible combinations of extreme points most of which are
infeasible, we select the fourth extreme point of $\mathbb{S}_{u_1}$, the second of
$\mathbb{S}_{u_2}$, the third of $\mathbb{S}_{u_3}$, and the first of $\mathbb{S}_{u_4}$.
For simplicity of notation, we show $J_{u_i}$ and $P_{U}(u_i)$ by $J_i$ and $P_i$, respectively.
 Thus, we have the following problem
\begin{align*}
\min\ &P_10.9544+P_20.6962+P_30.6254+P_40.9097\\+&P_1\epsilon J_1^2 7.3747+P_2\epsilon J_2^2 10.5816-P_3 \epsilon J_3^2 6.0808 \\
&\text{s.t.}\begin{bmatrix}
\frac{1}{2}\\\frac{1}{4}\\ \frac{1}{8}\\ \frac{1}{8}
\end{bmatrix} = P_1 \begin{bmatrix}
0\\0\\0.6251-10\epsilon J_1^2\\0.3749+10\epsilon J_1^2
\end{bmatrix}+ P_2\begin{bmatrix}
0.1875+5\epsilon J_2^2\\0\\0.8125-5\epsilon J_2^2\\0
\end{bmatrix}\\&+P_3\begin{bmatrix}
0\\0.1563-2.5\epsilon J_3^2\\0\\0.8437+2.5\epsilon J_3^2
\end{bmatrix}+P_4\begin{bmatrix}
0.675\\0.325\\0\\0
\end{bmatrix},\\
&P_1J_1^2+P_2J_2^2+P_3J_3^2=0,\ P_1,P_2,P_3,P_4\geq 0,\\
&|J_1^2|\leq \frac{1}{2},\ |J_2^2|\leq \frac{1}{2},\ |J_3^2|\leq \frac{1}{2},
\end{align*}  
where the minimization is over $P_U(u)$ and $J_u^2$, and $\epsilon=0.01$. Next, we convert the problem to a linear program. We have
\begin{align*}
\eta_1&= \begin{bmatrix}
0.6251P_1-10\epsilon P_1J_1^2\\0.3749P_1+10 \epsilon P_1J_1^2
\end{bmatrix}=\begin{bmatrix}
\eta_1^1\\\eta_1^2
\end{bmatrix},\\
\eta_2&=\begin{bmatrix}
0.1875P_2+ 5 \epsilon P_2J_2^2\\0.8125P_2- 5 \epsilon P_2J_2^2
\end{bmatrix}=\begin{bmatrix}
\eta_2^1\\\eta_2^2
\end{bmatrix},\\
\eta_3&= \begin{bmatrix}
0.1563P_3- 2.5 \epsilon P_3J_3^2\\0.8437P_3+ 2.5\epsilon P_3J_3^2
\end{bmatrix}=\begin{bmatrix}
\eta_3^1\\\eta_3^2
\end{bmatrix},\\
\eta_4&=P_4.
\end{align*}
The linear program is obtained as
\begin{align*}
\min\ &0.6779\eta_1^1+1.4155\eta_1^2+2.415\eta_2^1+0.2995\eta_2^2\\
&2.6776\eta_3^1+0.2452\eta_3^2+0.9097 \eta_4\\
&\text{s.t.}\begin{bmatrix}
\frac{1}{2}\\\frac{1}{4}\\ \frac{1}{8}\\ \frac{1}{8}
\end{bmatrix} = \begin{bmatrix}
0.675\eta_4+\eta_2^1\\0.325\eta_4+\eta_3^1\\\eta_2^2+\eta_1^1\\\eta_3^2+\eta_1^2
\end{bmatrix},\ \begin{cases}
\eta_1^1+\eta_1^2&\geq 0\\
\eta_2^1+\eta_2^2&\geq 0\\
\eta_3^1+\eta_3^2&\geq 0\\
\eta_4&\geq 0
\end{cases}\\
&\frac{0.6251\eta_1^2-0.3749\eta_1^1}{10}+\frac{0.8125\eta_2^1-0.1875\eta_2^2}{5}+\\
&\frac{0.1563\eta_3^2-0.8437\eta_3^1}{2.5}=0,\\
&\frac{|0.6251\eta_1^2-0.3749\eta_1^1|}{\eta_1^1+\eta_1^2}\!\leq 5\epsilon,\\ &\frac{|0.8125\eta_2^1-0.1875\eta_2^2|}{\eta_2^1+\eta_2^2}\!\leq 2.5\epsilon,\\
&\frac{|0.1563\eta_3^2-0.8437\eta_3^1|}{\eta_3^1+\eta_3^2}\leq 1.125\epsilon.
\end{align*}
This leads to $$P_U = \begin{bmatrix}
0 \\ 0.1488 \\ 0.143 \\ 0.7082
\end{bmatrix},$$
and we obtain $I(U;Y)\cong 0.9109$, which is less than the utility attained in \cite{Khodam22} and greater than that of the perfect-privacy case in \cite{borz}.
 This follows since, in this example, we force $u_4$ to satisfy perfect privacy. Furthermore, in contrast to the case of an invertible matrix $P_{X|Y}$, here, although $\epsilon_4 = 0$, we obtain a nonzero weight for $u_4$, i.e., $P_U(u_4) = 0.7082 > 0$. This justifies Remark~\ref{r4}.
\section{Conclusion}
In this paper, we studied an information-theoretic privacy mechanism design under different point-wise leakage constraints. We introduced the notion of multi-level point-wise leakage which generalizes existing point-wise leakage measures by allowing different leakage budgets for different disclosed outputs and also includes mixtures of perfect privacy and non-zero leakage. For small leakage regimes, we utilized information-geometric approximations to obtain tractable formulations in closed form of the privacy–utility trade-off. In particular, we showed that when the leakage matrix is invertible, the resulting problem follows a quadratic formulation with closed-form solutions, and that binary disclosed data is sufficient to achieve optimal utility. We further extended the approach to general leakage matrices via linear programming approximations and discussed methods for extending the permissible leakage range. These results provide a flexible and computationally efficient framework for privacy mechanism design with multi-level leakage control.
Finally, motivated by our model, several other directions can be considered. One interesting problem is to define a leakage measure over the alphabet of $X$, i.e., $\mathcal{L}(x_i \rightarrow U) \leq \epsilon_i$. A similar approach can be used in this case to obtain closed-form solutions. Furthermore, we may also consider the case in which the leakage from each symbol in the alphabet of $X$ to each symbol in the alphabet of $U$ is bounded, i.e., $\mathcal{L}(x_i \rightarrow u_j) \leq \epsilon_{i,j}$. In this case, we can consider measures such as LDP, LIP, and lift. We leave these directions for future work.

\clearpage   
\bibliographystyle{IEEEtran}
\bibliography{IEEEabrv,IZS}

\end{document}